\documentclass[12pt,a4paper]{article}
\usepackage[latin9]{inputenc}
\usepackage{amsmath}
\usepackage{amssymb}

\makeatletter

\newcommand{\lyxaddress}[1]{
	\par {\raggedright #1
	\vspace{1.4em}
	\noindent\par}
}

\@ifundefined{date}{}{\date{}}

\usepackage[active]{srcltx}

\usepackage{amsfonts}
\usepackage{amsthm}
\usepackage{mathrsfs}

\usepackage{relsize}

\newtheorem{theorem}{Theorem}
\newtheorem{proposition}[theorem]{Proposition}
\newtheorem{lemma}[theorem]{Lemma}
\newtheorem{corollary}[theorem]{Corollary}
\newtheorem*{lemma*}{Lemma}

\theoremstyle{remark}
\newtheorem{remark}[theorem]{Remark}
\newtheorem*{remark*}{Remark}
\newtheorem*{remarks*}{Remarks}
\newtheorem*{example*}{Example}
\newtheorem*{question*}{QUESTION}
\newtheorem*{conjecture*}{CONJECTURE}

\theoremstyle{definition}
\newtheorem{definition}[theorem]{Definition}
\newtheorem*{definition*}{Definition}

\newtheorem*{notation*}{Notation}

\newcommand{\Spec}{\mathop\mathrm{spec}\nolimits}
\newcommand{\supp}{\mathop\mathrm{supp}\nolimits}
\renewcommand{\span}{\mathop\mathrm{span}\nolimits}

\newcommand{\diag}{\mathop\mathrm{diag}\nolimits}
\newcommand{\Tr}{\mathop\mathrm{tr}\nolimits}

\makeatother

\begin{document}
\title{A family of orthogonal polynomials corresponding to Jacobi matrices
with a trace class inverse}
\author{Pavel \v{S}\v{t}ov\'\i\v{c}ek}
\maketitle

\lyxaddress{Department of Mathematics, Faculty of Nuclear Science, Czech Technical
University in Prague, Trojanova 13, 120~00 Praha, Czech Republic}
\begin{abstract}
\noindent Assume that $\{a_{n};\,n\geq0\}$ is a sequence of positive
numbers and $\sum a_{n}^{\,-1}<\infty$. Let $\alpha_{n}=ka_{n}$,
$\beta_{n}=a_{n}+k^{2}a_{n-1}$ where $k\in(0,1)$ is a parameter,
and let $\{P_{n}(x)\}$ be an orthonormal polynomial sequence defined
by the three-term recurrence
\[
\alpha_{0}P_{1}(x)+(\beta_{0}-x)P_{0}(x)=0,\ \alpha_{n}P_{n+1}(x)+(\beta_{n}-x)P_{n}(x)+\alpha_{n-1}P_{n-1}(x)=0
\]

\noindent for $n\geq1$, with $P_{0}(x)=1$. Let $J$ be the corresponding
Jacobi (tridiagonal) matrix, i.e. $J_{n,n}=\beta_{n}$, $J_{n,n+1}=J_{n+1,n}=\alpha_{n}$
for $n\geq0$. Then $J^{-1}$ exists and belongs to the trace class.
We derive an explicit formula for $P_{n}(x)$ as well as for the characteristic
function of $J$ and describe the orthogonality measure for the polynomial
sequence. As a particular case, the modified $q$-Laguerre polynomials
are introduced and studied.
\end{abstract}
\noindent \begin{flushleft}
\par\end{flushleft}

\section{Introduction}

A semi-infinite Jacobi (tridiagonal) matrix will be written in the
form
\begin{equation}
\mathcal{J}=\left(\begin{array}{ccccccc}
\beta_{0} & \alpha_{0}\\
\alpha_{0} & \beta_{1} & \alpha_{1}\\
 & \alpha_{1} & \beta_{2} & \alpha_{2}\\
 &  & \ddots & \ddots & \ddots\\
 &  &  & \ddots & \ddots\\
\\
\end{array}\right)\!.\label{eq:J}
\end{equation}
$\mathcal{J}$ is always assumed to be real and non-decomposable,
i.e. $\alpha_{n}\neq0$ for all $n\geq0$. Suppose we are given a
sequence of positive numbers $\{a_{n};\,n\geq0\}$ such that
\begin{equation}
\sum_{n=0}^{\infty}\frac{1}{a_{n}}<\infty,\label{eq:1-an-conv}
\end{equation}
and a parameter $k\in(0,1)$. We will focus on sequences $\{\alpha_{n}\}$
and $\{\beta_{n}\}$ defined as
\begin{equation}
\alpha_{n}:=ka_{n},\ \beta_{n}:=a_{n}+k^{2}a_{n-1}\ \,\text{for}\ n\in\mathbb{Z}_{+}\label{eq:alpha-beta-n}
\end{equation}
(with $\mathbb{Z}_{+}$ standing for non-negative integers).

We aim to study the orthonormal polynomial sequence $\{P_{n}(x);\,n\geq0\}$
defined by the three-term recurrence with initial data:
\begin{eqnarray}
 &  & \alpha_{0}P_{1}(x)+(\beta_{0}-x)P_{0}(x)=0,\nonumber \\
 &  & \alpha_{n}P_{n+1}(x)+(\beta_{n}-z)P_{n}(x)+\alpha_{n-1}P_{n-1}(x)=0\ \,\text{for}\ n\geq1,\label{eq:Pn-recurr}\\
 &  & P_{0}(x)=1.\nonumber 
\end{eqnarray}
Our goal is derivation of a formula for the polynomials $P_{n}(x)$
and a description of the respective orthogonality measure. This task
is closely related to analysis of spectral properties of an operator
in the Hilbert space $\ell^{2}(\mathbb{Z}_{+})$ represented by the
Jacobi matrix (\ref{eq:J}), (\ref{eq:alpha-beta-n}).

A short remark concerning the matrix operator is worthwhile. Let us
denote the canonical basis in $\ell^{2}(\mathbb{Z}_{+})$ as $\{\pmb{e}_{n};\,n\geq0\}$.
The matrix (\ref{eq:J}) clearly represents a symmetric operator in
$\ell^{2}(\mathbb{Z}_{+})$, with the domain being equal to the linear
hull of $\{\pmb{e}_{n}\}$. The symmetric operator will be called
$\stackrel{\mbox{{\ }}_{\circ}}{J}$. The deficiency indices of $\stackrel{\mbox{{\ }}_{\circ}}{J}$
are either $(0,0)$ or $(1,1)$. As is well known, this happens if
and only if the Hamburger moment problem for $\{P_{n}(x)\}$ is or
is not determinate, respectively \cite{Akhiezer}. Under our assumptions
it will turn out that $\stackrel{\mbox{{\ }}_{\circ}}{J}$ is essentially
self-adjoint. We will denote its self-adjoint closure by the symbol
$J$.

In order to formulate the result we introduce two complex functions
\begin{eqnarray}
\mathfrak{F}(z) & := & 1+\sum_{m=1}^{\infty}(-1)^{m}\label{eq:Fgt}\\
 &  & \times\bigg(\,\sum_{0\leq j_{1}<j_{2}<\ldots<j_{m}<\infty}\frac{(1-k^{2(j_{1}+1)})(1-k^{2(j_{2}-j_{1})})\cdots\left(1-k^{2(j_{m}-j_{m-1})}\right)}{(1-k^{2})^{m}a_{j_{1}}a_{j_{2}}\cdots a_{j_{m}}}\bigg)z^{m}\nonumber 
\end{eqnarray}
and
\begin{eqnarray}
\mathfrak{W}(z) & := & \sum_{m=0}^{\infty}(-1)^{m}\label{eq:Wgt}\\
 &  & \hskip-1.4em\times\bigg(\,\sum_{0\leq j_{0}<j_{1}<\ldots<j_{m}<\infty}\frac{k^{2j_{0}}(1-k^{2(j_{1}-j_{0})})(1-k^{2(j_{2}-j_{1})})\cdots(1-k^{2(j_{m}-j_{m-1})})}{(1-k^{2})^{m}a_{j_{0}}a_{j_{1}}a_{j_{2}}\ldots a_{j_{m}}}\bigg)z^{m}\nonumber 
\end{eqnarray}
With the above assumptions it is readily seen that the functions are
both entire. Moreover, we extend definition (\ref{eq:Wgt}) by defining
a countable family of entire functions, indexed by $n\in\mathbb{Z}_{+}$,
\begin{eqnarray}
\mathfrak{W}_{n}(z) & := & \sum_{m=0}^{\infty}(-1)^{m}\label{eq:Wgt-n}\\
 &  & \hskip-1.8em\times\bigg(\,\sum_{n\leq j_{0}<j_{1}<\ldots<j_{m}<\infty}\frac{k^{2j_{0}}(1-k^{2(j_{1}-j_{0})})(1-k^{2(j_{2}-j_{1})})\cdots(1-k^{2(j_{m}-j_{m-1})})}{(1-k^{2})^{m}a_{j_{0}}a_{j_{1}}a_{j_{2}}\ldots a_{j_{m}}}\bigg)z^{m}.\nonumber 
\end{eqnarray}
Hence $\mathfrak{W}_{0}(z)\equiv\mathfrak{W}(z)$. It is straightforward
to derive the estimate
\begin{equation}
|\mathfrak{W}_{n}(z)|\leq\frac{k^{2n}}{\min\{a_{j};\,j\geq n\}}\,\exp\!\bigg(\,\frac{|z|}{1-k^{2}}\,\sum_{j=n+1}^{\infty}\frac{1}{a_{j}}\bigg).\label{eq:Wn-estim}
\end{equation}

The most essential properties of $J$ and $\{P_{n}(x)\}$ are described
in the following two theorems.

\begin{theorem}\label{thm:J} Let $\{a_{n};\,n\geq0\}$ be a sequence
of positive numbers satisfying (\ref{eq:1-an-conv}) and\linebreak{}
$k\in(0,1)$. Furthermore, let $\{\alpha_{n}\}$ and $\{\beta_{n}\}$
be the sequences introduced in (\ref{eq:alpha-beta-n}). Then the
symmetric operator $\stackrel{\mbox{{\ }}_{\circ}}{J}$ in $\ell^{2}(\mathbb{Z}_{+})$
which is defined on the linear hull of $\{\pmb{e}_{n};\,n\geq0\}$
and represented by the Jacobi matrix (\ref{eq:J}) is essentially
self-adjoint. The spectrum of its closure, a self-adjoint operator
$J$, satisfies
\begin{equation}
\Spec(J)=\Spec_{p}(J)=\{\lambda_{j};\,j\in\mathbb{Z}_{+}\},\label{eq:spec-J}
\end{equation}
with all eigenvalues $\lambda_{j}$ being positive and simple. Moreover,
it holds true that
\[
\sum_{j=0}^{\infty}\frac{1}{\lambda_{j}}=\sum_{j=0}^{\infty}\frac{1-k^{2j+2}}{(1-k^{2})a_{j}}<\infty
\]
so that $J^{-1}$ exists and belongs to the trace class. For every
$j\in\mathbb{Z}_{+}$, an eigenvector corresponding to the eigenvalue
$\lambda_{j}$ can be chosen as the column vector
\begin{equation}
\big(\Phi_{0}(\lambda_{j}),\Phi_{1}(\lambda_{j}),\Phi_{2}(\lambda_{j}),\ldots\big)^{T}\in\ell^{2}(\mathbb{Z}_{+}),\ \text{where}\ \Phi_{n}(z):=(-1)^{n}k^{-n}\mathfrak{W}_{n}(z),\label{eq:Phin}
\end{equation}
with $\mathfrak{W}_{n}(z)$ defined in (\ref{eq:Wgt-n}). The entire
function $\mathfrak{F}(z)$ defined in (\ref{eq:Fgt}) is a characteristic
function of $J$ in the sense that, for all $z\in\mathbb{C}$, 
\begin{equation}
\mathfrak{F}(z)=\prod_{j=0}^{\infty}\bigg(1-\frac{z}{\lambda_{j}}\bigg).\label{eq:Fgt-prod}
\end{equation}
\end{theorem}

\begin{theorem}\label{thm:P} Let $\{a_{n};\,n\geq0\}$ be a sequence
of positive numbers satisfying (\ref{eq:1-an-conv}), and $k\in(0,1)$.
Furthermore, let $\{P_{n}(x);\,n\geq0\}$ be the orthonormal polynomial
sequence defined in (\ref{eq:alpha-beta-n}), (\ref{eq:Pn-recurr}).
Then, for all $n\in\mathbb{Z}_{+}$,
\begin{eqnarray}
(-1)^{n}k^{n}P_{n}(x) & = & 1+\sum_{m=1}^{n}(-1)^{m}\label{eq:Pn}\\
 &  & \hskip-4em\times\bigg(\,\sum_{0\leq j_{1}<j_{2}<\ldots<j_{m}\leq n-1}\frac{(1-k^{2(j_{1}+1)})(1-k^{2(j_{2}-j_{1})})\cdots(1-k^{2(j_{m}-j_{m-1})})}{(1-k^{2})^{m}a_{j_{1}}a_{j_{2}}\cdots a_{j_{m}}}\bigg)x^{m}.\nonumber 
\end{eqnarray}
The Hamburger moment problem for $\{P_{n}(x)\}$ is determinate. The
unique orthogonality measure is supported on the zero set of the function
$\mathfrak{F}(z)$ defined in (\ref{eq:Fgt}), i.e. the measure support
coincides with the spectrum $\{\lambda_{j};\,j\geq0\}$ of the operator
$J$, as described in Theorem~\ref{thm:J}. It holds true that
\[
\forall s,t\in\mathbb{Z}_{+},\ \sum_{j=0}^{\infty}\mu_{j}P_{s}(\lambda_{j})P_{t}(\lambda_{j})=\delta_{s,t}
\]
where the masses $\mu_{j}$, $j\in\mathbb{Z}_{+}$, are given by
\begin{equation}
\mu_{j}=-\frac{\mathfrak{W}(\lambda_{j})}{\mathfrak{F}'(\lambda_{j})},\label{eq:mass-muj}
\end{equation}
with $\mathfrak{W}(z)$ being defined in (\ref{eq:Wgt}). \end{theorem}

\begin{remark} Let us mention an interpretation of the function $\mathfrak{W}(z)$
which is related to the notion of the associated Jacobi matrix. By
definition, the associated Jacobi matrix $\mathcal{J}^{(1)}$ is obtained
from $\mathcal{J}$ given in (\ref{eq:J}) by deleting the first row
and column. Under the same assumption as in Theorem \ref{thm:J},
there exists exactly one self-adjoint operator $J^{(1)}$ whose matrix
in the canonical basis equals $\mathcal{J}^{(1)}$, $J^{(1)}$ is
positive definite and its inverse belongs to the trace class, see
Proposition \ref{thm:associated} below. $\mathfrak{W}(z)$ is a characteristic
function of $J^{(1)}$ in the sense that, for all $z\in\mathbb{C}$,
\[
\mathfrak{W}(z)=\mathfrak{W}(0)\prod_{j=0}^{\infty}\bigg(1-\frac{z}{\lambda_{j}^{(1)}}\bigg)
\]
where $\Spec(J^{(1)})=\Spec_{p}(J^{(1)})=\{\lambda_{j}^{(1)};\,j\in\mathbb{Z}_{+}\}$.
The proof of this fact is omitted, however, since it is not substantial
for the rest of the paper. \end{remark}

The remainder of the paper is devoted to a proof of Theorems \ref{thm:J},
\ref{thm:P} and to an application. The paper is organized as follows.
In Section \ref{sec:Preliminaries} we provide a brief summary of
some basic notions and preparatory results. Sections \ref{sec:Proofs}
contains a proof of Theorems \ref{thm:J} and \ref{thm:P}. This task
is accomplished by proving a series of separate propositions and lemmas.
In Section \ref{sec:Auxiliary} we derive some summation formulas
which are needed in the following section, Section \ref{sec:qLaguerre},
in which a modification of $q$-Laguerre polynomials is proposed.
It turns out that the orthogonality measure for the modified $q$-Laguerre
polynomials is supported on the roots of a Jackson $q$-Bessel function
of the second kind.

\section{Preliminaries\label{sec:Preliminaries}}

It can be shown that if there exists exactly one self-adjoint operator
$J$ whose matrix in the canonical basis equals $\mathcal{J}$ then
a similar assertion is true for the associated Jacobi matrix $\mathcal{J}^{(1)}$
as well \cite{Stovicek-JAT}. The self-adjoint operator corresponding
to $\mathcal{J}^{(1)}$ will be denoted $J^{(1)}$.

\begin{proposition}\label{thm:associated} Assume that there exists
exactly one self-adjoint operator $J$ in $\ell^{2}(\mathbb{Z}_{+})$
whose matrix in the canonical basis equals a Jacobi matrix $\mathcal{J}$.
Assume further that $J$ is bounded below by a positive constant and
$J^{-1}$ belongs to the trace class. Then the same is true for the
self-adjoint operator $J^{(1)}$ corresponding to the associated Jacobi
matrix $\mathcal{J}^{(1)}$. \end{proposition}

\begin{proof} The proof is computational and is outlined below. From
the fact that the matrix of $J$ in the canonical basis is tridiagonal
it follows that
\begin{eqnarray*}
\beta_{0}\langle e_{0},J^{-1}\pmb{e}_{0}\rangle+\alpha_{0}\langle\pmb{e}_{1},J^{-1}\pmb{e}_{0}\rangle & = & 1,\\
\beta_{0}^{\,2}\,\langle\pmb{e}_{0},J^{-1}\pmb{e}_{0}\rangle-\alpha_{0}^{\,2}\,\langle\pmb{e}_{1},J^{-1}\pmb{e}_{1}\rangle & = & \beta_{0}.
\end{eqnarray*}
Using these relations it is straightforward to show that for
\begin{eqnarray*}
M & := & \left(\begin{array}{cc}
1 & 0\\
0 & 1
\end{array}\right)-\alpha_{0}\left(\begin{array}{cc}
\langle\pmb{e}_{1},J^{-1}\pmb{e}_{0}\rangle & \langle\pmb{e}_{1},J^{-1}\pmb{e}_{1}\rangle\\
\langle\pmb{e}_{0},J^{-1}\pmb{e}_{0}\rangle & \langle\pmb{e}_{1},J^{-1}\pmb{e}_{0}\rangle
\end{array}\right)\\
\noalign{\smallskip} & = & \left(\begin{array}{cc}
\beta_{0}\,\langle\pmb{e}_{0},J^{-1}\pmb{e}_{0}\rangle & -\alpha_{0}\,\langle\pmb{e}_{1},J^{-1}\pmb{e}_{1}\rangle\\
-\alpha_{0}\,\langle\pmb{e}_{0},J^{-1}\pmb{e}_{0}\rangle & \beta_{0}\,\langle\pmb{e}_{0},J^{-1}\pmb{e}_{0}\rangle
\end{array}\right)
\end{eqnarray*}
we have
\[
\det(M)=\beta_{0}\,\langle\pmb{e}_{0},J^{-1}\pmb{e}_{0}\rangle.
\]

Denote as $\beta_{0}\oplus J^{(1)}$ the block diagonal matrix with
blocks $(\beta_{0})$ and $J^{(1)}$. Equation
\[
(\beta_{0}\oplus J^{(1)})\!\bigg(J^{-1}+J^{-1}\bigg(\,\sum_{s=0}^{1}\sum_{t=0}^{1}x_{s,t}\,\pmb{e}_{s}\pmb{e}_{t}{}^{T}\bigg)J^{-1}\bigg)\!=I
\]
results in four equations for four unknowns $x_{s,t}$, $0\leq s,t\leq1$,
written in the matrix form
\[
\alpha_{0}\!\left(\begin{array}{cc}
0 & 1\\
1 & 0
\end{array}\right)=\left(\left(\begin{array}{cc}
1 & 0\\
0 & 1
\end{array}\right)-\alpha_{0}\!\left(\begin{array}{cc}
\left\langle \pmb{e}_{1},J^{-1}\pmb{e}_{0}\right\rangle  & \left\langle \pmb{e}_{1},J^{-1}\pmb{e}_{1}\right\rangle \\
\left\langle \pmb{e}_{0},J^{-1}\pmb{e}_{0}\right\rangle  & \left\langle \pmb{e}_{1},J^{-1}\pmb{e}_{0}\right\rangle 
\end{array}\right)\right)\!\left(\begin{array}{cc}
x_{0,0} & x_{0,1}\\
x_{1,0} & x_{1,1}
\end{array}\right)\!.
\]
A unique solution is readily obtained. Since the inverse is unique
we have
\[
(\beta_{0}\oplus J^{(1)})^{-1}=J^{-1}+J^{-1}\bigg(\,\sum_{s=0}^{1}\sum_{t=0}^{1}x_{s,t}\,\pmb{e}_{s}\pmb{e}_{t}{}^{T}\bigg)J^{-1}
\]
where
\[
x_{0,0}=\frac{\alpha_{0}^{\,2}\langle\pmb{e}_{1},J^{-1}\pmb{e}_{1}\rangle}{\beta_{0}\langle\pmb{e}_{0},J^{-1}\pmb{e}_{0}\rangle}\,,\ x_{0,1}=x_{1,0}=\alpha_{0},\ x_{1,1}=\frac{\alpha_{0}^{\,2}}{\beta_{0}}\,.
\]
Whence
\[
\Tr\big((J^{(1)})^{-1}\big)=\Tr(J^{-1})+\sum_{s=0}^{1}\sum_{t=0}^{1}x_{s,t}\langle\pmb{e}_{s},J^{-2}\,\pmb{e}_{t}\rangle-\frac{1}{\beta_{0}}.
\]
This shows the proposition. \end{proof}

We still assume that $J$ is a unique self-adjoint operator in $\ell^{2}(\mathbb{Z}_{+})$
whose matrix in the canonical basis equals $\mathcal{J}$, as given
in (\ref{eq:J}). Let $\mu$ be the unique Borel probability measure
on $\mathbb{R}$ which is an orthogonality measure for the orthonormal
polynomial sequence $\{P_{n}(x)\}$ defined in (\ref{eq:Pn-recurr}).
Furthermore, let $w_{n}(z)$, $n\geq0$, be the so called functions
of the second kind \cite{VanAssche},
\begin{equation}
\forall z\in\varrho(J),\,\forall n\geq0,\ w_{n}(z):=\int\frac{P_{n}(\lambda)}{\lambda-z}\,\text{d}\mu(\lambda).\label{eq:fce-2nd-def}
\end{equation}
Here $\varrho(J)$ denotes the resolvent set of $J$. In particular,
\begin{equation}
w(z)\equiv w_{0}(z)=\int\frac{\text{d}\mu(\lambda)}{\lambda-z}\label{eq:Weyl-def}
\end{equation}
is the Weyl function of $J$. Thus the Weyl function is the Stieltjes
transformation of the measure $\mu$.

One can check that \cite{Stovicek-JAT}
\[
\forall z\in\varrho(J),\,\forall n\geq0,\ w_{n}(z)P_{n}(z)=\langle\pmb{e}_{n},(J-z)^{-1}\pmb{e}_{n}\rangle.
\]
Particularly, if $J$ is bounded below by a positive constant $\gamma$
and so $0\in\varrho(J)$ then
\begin{equation}
\sum_{n=0}^{\infty}\langle\pmb{e}_{n},J^{-1}\pmb{e}_{n}\rangle=\sum_{n=0}^{\infty}w_{n}(0)P_{n}(0).\label{eq:tr-Jinv}
\end{equation}
Hence $J^{-1}$ belongs to the trace class if and only if this sum
converges and, if so, the sum equals $\Tr J^{-1}$. Moreover, the
support of the unique orthogonality measure is contained in $[\,\gamma,\infty)$.
Necessarily, all roots of the polynomials $P_{n}(z)$, $n\in\mathbb{N}$,
are contained in this interval, too \cite{Berg}. We have
\begin{equation}
\forall z\in\mathbb{C}\setminus[\,\gamma,\infty),\,\forall n\geq0,\ w_{n}(z)=-\bigg(\sum_{j=n}^{\infty}\frac{1}{\alpha_{j}P_{j}(z)P_{j+1}(z)}\bigg)P_{n}(z).\label{eq:wn}
\end{equation}

In addition, as also proven in \cite{Stovicek-JAT}, if $J$ is bounded
below by a positive constant and $J^{-1}$ is a trace class operator
then
\begin{equation}
\mathfrak{F}_{\text{char}}(z):=1-z\sum_{n=0}^{\infty}w_{n}(0)P_{n}(z)\label{eq:Fchar}
\end{equation}
is a characteristic function of $J$ in the sense that
\[
\mathfrak{F}_{\text{char}}(z)=\prod_{j=0}^{\infty}\bigg(1-\frac{z}{\lambda_{j}}\bigg)\ \text{where}\ \Spec(J)=\Spec_{p}(J)=\{\lambda_{j};\,j\geq0\}.
\]

\section{Proofs of Theorems \ref{thm:J} and \ref{thm:P}\label{sec:Proofs}}

All the claims contained in Theorems \ref{thm:J} and \ref{thm:P}
will be proven step by step as separate lemmas or propositions while
assuming everywhere that $\{a_{n}\}$ is a sequence of positive numbers
satisfying (\ref{eq:1-an-conv}) and the sequences $\{\alpha_{n}\}$
and $\{\beta_{n}\}$ are given by equations (\ref{eq:alpha-beta-n}).
Furthermore, $\{P_{n}(x)\}$ denotes the orthonormal polynomial sequence
defined in (\ref{eq:Pn-recurr}).

\begin{proposition}\label{thm:J-trclass} The symmetric operator
$\stackrel{\mbox{{\ }}_{\circ}}{J}$ in $\ell^{2}(\mathbb{Z}_{+})$,
with the matrix $\mathcal{J}$ given in (\ref{eq:J}), is essentially
self-adjoint, and the Hamburger moment problem for $\{P_{n}(x)\}$
is determinate. The closure of $\stackrel{\mbox{{\ }}_{\circ}}{J}$,
i.e. the self-adjoint operator $J$, is positive definite and $J^{-1}$
belongs to the trace class. \end{proposition}

\begin{proof} From (\ref{eq:alpha-beta-n}), (\ref{eq:Pn-recurr})
we infer that the three-term recurrence in case $x=0$ can be reduced
to a two-term recurrence, namely we get
\[
ka_{n}\big(P_{n+1}(0)+k^{-1}P_{n}(0)\big)+k^{2}a_{n-1}\big(P_{n}(0)+k^{-1}P_{n-1}(0)\big)=0\ \,\text{for}\ n\geq1.
\]
Taking into account the initial data we have
\begin{equation}
P_{n}(0)=(-1)^{n}k^{-n}\ \,\text{for}\ n\geq0.\label{eq:Pn0}
\end{equation}
Thus the sum $\sum_{n}P_{n}(0)^{2}$ is divergent and the Hamburger
moment problem is determinate \cite{Akhiezer}. This in turn means
that the symmetric operator $\stackrel{\mbox{{\ }}_{\circ}}{J}$ is
essentially self-adjoint \cite{Akhiezer}.

The matrix $\mathcal{J}$ can be decomposed,
\[
\mathcal{J}=(I+kE)J_{0}(I+kE^{T}),
\]
where $I$ is the unite matrix, $J_{0}:=\diag(a_{0},a_{1},a_{2},\ldots)$
and $E$ is a strictly lower triangular matrix whose only nonzero
elements are $E_{n+1,n}=1$, $n\geq0$. Since all factors are band
matrices the product is well defined. It follows that
\[
\forall\pmb{f}\in\span\{\pmb{e}_{n}\},\ \langle\pmb{f},J\pmb{f}\rangle\geq a_{\text{min}}\,\|(I+kE^{T})\pmb{f}\|^{2}\geq a_{\text{min}}\,(1-k)^{2}\,\|\pmb{f}\|^{2}
\]
where
\[
a_{\text{min}}:=\min\{a_{n};\,n\geq0\}.
\]
Since $\span\{\pmb{e}_{n}\}$ is a core for $J$ we have
\[
J\geq\gamma\ \,\text{where}\ \gamma:=a_{\text{min}}\,(1-k)^{2}>0.
\]

In view of (\ref{eq:Pn0}) and recalling (\ref{eq:wn}), (\ref{eq:tr-Jinv})
we get
\begin{equation}
w_{n}(0)=(-1)^{n}\sum_{j=n}^{\infty}\frac{k^{2j-n}}{a_{j}}=(-1)^{n}k^{n}\sum_{j=0}^{\infty}\frac{k^{2j}}{a_{n+j}}\label{eq:wn0}
\end{equation}
and
\begin{eqnarray*}
\Tr J^{-1}=\sum_{n=0}^{\infty}w_{n}(0)P_{n}(0) & = & -\sum_{n=0}^{\infty}P_{n}(0)^{2}\sum_{j=n}^{\infty}\frac{1}{\alpha_{j}P_{j}(0)P_{j+1}(0)}\,=\,\sum_{n=0}^{\infty}k^{-2n}\sum_{j=n}^{\infty}\frac{k^{2j}}{a_{j}}\\
 & = & \sum_{j=0}^{\infty}\frac{1}{a_{j}}\sum_{n=0}^{j}k^{2(j-n)}\,=\,\sum_{j=0}^{\infty}\frac{1-k^{2j+2}}{(1-k^{2})a_{j}}\,<\infty.
\end{eqnarray*}
Hence $J^{-1}$ belongs to the trace class \end{proof}

\begin{remark}\label{rem:tr-Jinv} We have derived that
\[
\Tr J^{-1}=\sum_{j=0}^{\infty}\frac{1-k^{2j+2}}{(1-k^{2})a_{j}}\,.
\]
\end{remark}

\begin{proposition}\label{thm:Pn} The polynomial sequence $\{P_{n}(x)\}$
defined in (\ref{eq:Pn-recurr}) satisfies (\ref{eq:Pn}).

\end{proposition}

\begin{proof} Denote by $\pi_{n}(x)$ the RHS of (\ref{eq:Pn}).
We have to verify that the sequence $\{(-1)^{n}k^{-n}\pi_{n}(x)\}$
satisfies (\ref{eq:Pn-recurr}). The initial condition is actually
fulfilled since $\pi_{0}(x)=1$. Furthermore, $\pi_{1}(x)=1-x/a_{0}$,
and thus the beginning of the recurrence for $n=0$ is immediately
seen to be satisfied, too. So we can focus on the recurrence for $n\geq1$.

In view of (\ref{eq:alpha-beta-n}), the recurrence we wish to prove
takes the form
\begin{equation}
a_{n}\big(\pi_{n+1}(x)-\pi_{n}(x)\big)-k^{2}a_{n-1}\big(\pi_{n}(x)-\pi_{n-1}(x)\big)=-x\pi_{n}(x).\label{eq:pin-reccur}
\end{equation}
Note that in the outer sum on the RHS of (\ref{eq:Pn}) we can write
$\sum_{m=1}^{\infty}$ instead of $\sum_{m=1}^{n}$ owing to the constraint
on the indices $j_{i}$ in the inner sum. Considering the formula
for $\pi_{n+1}(x)$ and separating the cases $j_{m}<n$ and $j_{m}=n$
we obtain
\begin{eqnarray*}
\pi_{n+1}(x) & = & \pi_{n}(x)+\frac{(1-k^{2(n+1)})x}{(1-k^{2})a_{n}}+\sum_{m=2}^{\infty}(-1)^{m}x^{m}\\
 &  & \hskip-6em\times\sum_{0\leq j_{1}<j_{2}<\ldots<j_{m-1}<n}\frac{(1-k^{2(j_{1}+1)})(1-k^{2(j_{2}-j_{1})})\cdots(1-k^{2(j_{m-1}-j_{m-2})})(1-k^{2(n-j_{m-1})})}{(1-k^{2})^{m}a_{j_{1}}a_{j_{2}}\cdots a_{j_{m-1}}a_{n}}.
\end{eqnarray*}
This can be rewritten as
\begin{eqnarray*}
a_{n}\big(\pi_{n+1}(x)-\pi_{n}(x)\big) & = & x\,\frac{1-k^{2(n+1)}}{1-k^{2}}-x\sum_{m=1}^{\infty}(-1)^{m}\\
 &  & \hskip-6em\times\bigg(\,\sum_{0\leq j_{1}<j_{2}<\ldots<j_{m}\leq n-1}\frac{(1-k^{2(j_{1}+1)})(1-k^{2(j_{2}-j_{1})})\cdots(1-k^{2(j_{m}-j_{m-1})})}{(1-k^{2})^{m}a_{j_{1}}a_{j_{2}}\cdots a_{j_{m}}}\\
 &  & \qquad\qquad\times\,\frac{1-k^{2(n-j_{m})}}{1-k^{2}}\bigg)x^{m}.
\end{eqnarray*}
A similar formula holds for $a_{n-1}\big(\pi_{n}(x)-\pi_{n-1}(x)\big)$
where the constraint $\ldots<j_{m}\leq n-2$ can be replaced by $\ldots<j_{m}\leq n-1$
owing to the factor $(1-k^{2(n-1-j_{m})})$. Now it is straightforward
to evaluate the LHS of (\ref{eq:pin-reccur}) and consequently to
find out that equation (\ref{eq:pin-reccur}) is actually valid. \end{proof}

\begin{proposition}\label{thm:Fchar} The function $\mathfrak{F}(z)$
defined in (\ref{eq:Fgt}) is a characteristic function of $J$ in
the sense that equations (\ref{eq:spec-J}), (\ref{eq:Fgt-prod})
hold. \end{proposition}

\begin{proof} As explained in Section \ref{sec:Preliminaries}, we
have to show that the function $\mathfrak{F}_{\text{char}}(z)$ defined
in (\ref{eq:Fchar}) equals the entire function $\mathfrak{F}(z)$.
To this end we can use the already proven formulas (\ref{eq:Pn})
and (\ref{eq:wn0}). Thus we can evaluate the coefficient at $z^{m+1}$,
$m\in\mathbb{N}$, on the RHS of (\ref{eq:Fchar}). For an $m$-tuple
of indices $j_{1},\ldots,j_{m}\in\mathbb{Z}_{+}$ let us denote
\[
s(j_{1},\ldots,j_{m}):=\frac{\left(1-k^{2(j_{1}+1)}\right)\left(1-k^{2(j_{2}-j_{1})}\right)\ldots\left(1-k^{2(j_{m}-j_{m-1})}\right)}{(1-k^{2})^{m}a_{j_{1}}a_{j_{2}}\ldots a_{j_{m}}}.
\]
Then the coefficient at $z^{m+1}$ equals, up to the sign $(-1)^{m+1}$,
\begin{eqnarray*}
 &  & \sum_{n=m}^{\infty}\sum_{j=0}^{\infty}\frac{k^{2j}}{a_{n+j}}\,\sum_{0\leq j_{1}<j_{2}<\ldots<j_{m}\leq n-1}s(j_{1},\ldots,j_{m})\\
 &  & =\,\sum_{0\leq j_{1}<j_{2}<\ldots<j_{m}<\infty}\,\,\sum_{n=j_{m}+1}^{\infty}\sum_{j=0}^{\infty}\frac{k^{2j}}{a_{n+j}}\,s(j_{1},\ldots,j_{m})\\
 &  & =\sum_{0\leq j_{1}<j_{2}<\ldots<j_{m}<\infty}\,\,\sum_{j_{m+1}=j_{m}+1}^{\infty}\sum_{j=0}^{j_{m+1}-j_{m}-1}k^{2j}\,\frac{s(j_{1},\ldots,j_{m})}{a_{j_{m+1}}}\\
 &  & =\,\sum_{0\leq j_{1}<j_{2}<\ldots<j_{m}<j_{m+1}<\infty}\frac{s(j_{1},\ldots,j_{m})\,(1-k^{2(j_{m+1}-j_{m})})}{(1-k^{2})a_{j_{m+1}}},
\end{eqnarray*}
and this is in agreement with (\ref{eq:Fgt}). One can proceed similarly
for the coefficient standing at $z$. This coefficient equals, again
up to a sign,
\[
\sum_{n=0}^{\infty}\sum_{j=0}^{\infty}\frac{k^{2j}}{a_{n+j}}=\sum_{j_{1}=0}^{\infty}\sum_{j=0}^{j_{1}}\frac{k^{2j}}{a_{j_{1}}}=\sum_{j_{1}=0}^{\infty}\frac{1-k^{2(j_{1}+1)}}{(1-k^{2})a_{j_{1}}}\,.
\]
This shows that $\mathfrak{F}(z)=\mathfrak{F}_{\text{char}}(z)$.
\end{proof}

\begin{lemma}\label{thm:Phin} The sequence of functions $\{\Phi_{n}(z)\}$
defined in (\ref{eq:Phin}) and (\ref{eq:Wgt-n}) obeys the recurrence
equation
\begin{equation}
\alpha_{n}\Phi_{n+1}(z)+(\beta_{n}-z)\Phi_{n}(z)+\alpha_{n-1}\Phi_{n-1}(z)=0\text{ }\ \text{for}\ n\geq1,\label{eq:Phi-recurr}
\end{equation}
and also
\begin{equation}
\alpha_{0}\Phi_{1}(z)+(\beta_{0}-z)\Phi_{0}(z)-\mathfrak{F}(z)=0.\label{eq:Phi-recurr0}
\end{equation}
\end{lemma}

\begin{proof} From (\ref{eq:Wgt-n}) it is seen that
\begin{eqnarray*}
\mathfrak{W}_{n}(z)-\mathfrak{W}_{n+1}(z) & = & \frac{k^{2n}}{a_{n}}\,\sum_{m=0}^{\infty}(-1)^{m}\\
 &  & \hskip-5.2em\times\bigg(\sum_{n+1\leq j_{1}<j_{2}<\ldots<j_{m}<\infty}\!\frac{(1-k^{2(j_{1}-n)})(1-k^{2(j_{2}-j_{1})})\cdots(1-k^{2(j_{m}-j_{m-1})})}{(1-k^{2})^{m}a_{j_{1}}a_{j_{2}}\ldots a_{j_{m}}}\bigg)z^{m}.
\end{eqnarray*}
Here again, the constraint in the inner sum , $n+1\leq j_{1}<\ldots$,
can be replaced by $n\leq j_{1}<\ldots$ because of the factor $(1-k^{2(j_{1}-n)})$.
We have a similar expression for $\mathfrak{W}_{n-1}(z)-\mathfrak{W}_{n}(z)$.
Now one can readily verify that
\[
a_{n}\big(\mathfrak{W}_{n}(z)-\mathfrak{W}_{n+1}(z)\big)-k^{2}a_{n-1}\big(\mathfrak{W}_{n-1}(z)-\mathfrak{W}_{n}(z)\big)=z\,\mathfrak{W}_{n}(z).
\]
In regard of (\ref{eq:alpha-beta-n}) and (\ref{eq:Phin}), this equation
is equivalent to (\ref{eq:Phi-recurr}).

Let us extend the sequence $\{a_{n};\,n\geq0\}$ by prepending to
it an additional member $a_{-1}$, $0<a_{-1}<a_{\text{min}}$. Then
one can define functions $\mathfrak{W}_{-1}(z)$ and $\Phi_{-1}(z)$
correspondingly, just by extending definitions (\ref{eq:Wgt-n}),
(\ref{eq:Phin}) to $n=-1$. With this extended definition, equation
(\ref{eq:Phi-recurr}) holds also for $n=0$. It is immediate to see
that
\begin{equation}
\lim_{a_{-1}\to0+}a_{-1}\,\mathfrak{W}_{-1}(z)=k^{-2}\mathfrak{F}(z)\ \,\text{whence}\ \lim_{a_{-1}\to0+}ka_{-1}\,\Phi_{-1}(z)=-\mathfrak{F}(z).\label{eq:lim-a-1}
\end{equation}
Letting $n=0$ in (\ref{eq:Phi-recurr}) and applying this limit we
obtain (\ref{eq:Phi-recurr0}). \end{proof}

\begin{lemma} Let $\{\Phi_{n}(z)\}$ be the sequence defined in (\ref{eq:Phin}),
(\ref{eq:Wgt-n}). The Wronskian of the sequences $\{P_{n}(z)\}$
and $\{\Phi_{n}(z)\}$ is a constant and fulfills
\begin{equation}
\alpha_{n}\big(P_{n}(z)\Phi_{n+1}(z)-P_{n+1}(z)\Phi_{n}(z)\big)=\mathfrak{F}(z)\ \,\text{for all}\ n\geq0.\label{eq:Wronskian}
\end{equation}
\end{lemma}

\begin{proof} The sequences $\{P_{n}(z)\}$ and $\{\Phi_{n}(z)\}$
satisfy both the same three-term recurrence and therefore their Wronskian
is a constant for $n\geq0$. To find the constant one can use the
same reasoning as in the proof of Lemma \ref{thm:Phin} and consider
an extended sequence of positive numbers $\{a_{n};\,n\geq-1\}$. Then
(\ref{eq:Phi-recurr}) is valid also for $n=0$. The second equation
in (\ref{eq:Pn-recurr}) is valid for $n=0$ as well provided we let
$P_{-1}(z)=0$. Thus the constant equals
\[
\alpha_{-1}\big(P_{-1}(z)\Phi_{0}(z)-P_{0}(z)\Phi_{-1}(z)\big)=-ka_{-1}\Phi_{-1}(z).
\]
Applying the limit (\ref{eq:lim-a-1}) in this expression one obtains
(\ref{eq:Wronskian}). \end{proof}

\begin{proposition}\label{thm:Phi-eigenvec} If $\lambda$ is an
eigenvalue of $J$ then the column vector
\[
\pmb{\Phi}(\lambda):=\big(\Phi_{0}(\lambda),\Phi_{1}(\lambda),\Phi_{2}(\lambda),\ldots\big)^{T}\in\ell^{2}(\mathbb{Z}_{+})
\]
is an eigenvector of $J$ corresponding to $\lambda$. Here $\{\Phi_{n}(z)\}$
is again the sequences defined in (\ref{eq:Phin}), (\ref{eq:Wgt-n}).
The norm of the eigenvector fulfills
\[
\|\pmb{\Phi}(\lambda)\|^{2}=-\mathfrak{F}'(\lambda)\mathfrak{W}(\lambda).
\]
\end{proposition}

\begin{proof} From (\ref{eq:Wn-estim}) and (\ref{eq:Phin}) it follows
that
\begin{equation}
\left|\Phi_{n}(z)\right|\leq\frac{k^{n}}{\min\{a_{j};\,j\geq n\}}\,\exp\!\bigg(\,\frac{|z|}{1-k^{2}}\,\sum_{j=n+1}^{\infty}\frac{1}{a_{j}}\bigg)\ \text{for}\ n\geq0,\label{eq:Phin-estim}
\end{equation}
and so the sequence $\{\Phi_{n}(z)\}$ is square summable for every
$z\in\mathbb{C}$. If $\lambda$ is an eigenvalue of $J$ then, by
Proposition \ref{thm:Fchar}, $\mathfrak{F}(\lambda)=0$ and, according
to (\ref{eq:Wronskian}), the sequences $\{P_{n}(\lambda)\}$ and
$\{\Phi_{n}(\lambda)\}$ are linearly dependent. From the recurrence
(\ref{eq:Pn-recurr}) one infers that $\pmb{\Phi}(\lambda)$ is an
eigenvector of $J$ corresponding to the eigenvalue $\lambda$.

Using again the extended sequence of positive numbers $\{a_{n};\,n\geq-1\}$,
as in the proof of Lemma \ref{thm:Phin}, and referring to the Cristoffel-Darboux
formula \cite{Akhiezer} we have, for $N\in\mathbb{N}$ and every
couple $\lambda,\theta\in\mathbb{C}$,
\begin{eqnarray*}
(\lambda-\theta)\sum_{n=0}^{N}\Phi_{n}(\lambda)\Phi_{n}(\theta) & = & \alpha_{-1}\big(\Phi_{-1}(\lambda)\Phi_{0}(\theta)-\Phi_{0}(\lambda)\Phi_{-1}(\theta)\big)\\
 &  & -\,\alpha_{N}\big(\Phi_{N}(\lambda)\Phi_{N+1}(\theta)-\Phi_{N+1}(\lambda)\Phi_{N}(\theta)\big).
\end{eqnarray*}
Owing to the estimate (\ref{eq:Phin-estim}) one can send $N\to\infty.$
In addition, one can again apply the limit (\ref{eq:lim-a-1}) to
this expression. This way we get
\[
(\lambda-\theta)\sum_{n=0}^{\infty}\Phi_{n}(\lambda)\Phi_{n}(\theta)=-\mathfrak{F}(\lambda)\Phi_{0}(\theta)+\Phi_{0}(\lambda)\mathfrak{F}(\theta).
\]
In the case considered $\lambda$ is an eigenvalue of $J$ and $\mathfrak{F}(\lambda)=0$.
Then in the limit $\theta\to\lambda$ we obtain
\begin{equation}
\sum_{n=0}^{\infty}\Phi_{n}(\lambda)^{2}=-\mathfrak{F}'(\lambda)\Phi_{0}(\lambda).\label{eq:sum-Phin}
\end{equation}
But note that $\Phi_{0}(z)=\mathfrak{W}(z)$. \end{proof}

\begin{corollary}\label{thm:mass} Let $\mu$ be the unique Borel
probability measure on $\mathbb{R}$ which is an orthogonality measure
for $\{P_{n}(x)\}$. Then for all eigenvalues $\lambda$ of $J$,
\[
\mu(\{\lambda\})=-\frac{\mathfrak{W}(\lambda)}{\mathfrak{F}'(\lambda)}.
\]
\end{corollary}

\begin{proof} According to the general theory \cite[Subsec. 2.5]{Akhiezer}
and by Proposition \ref{thm:Fchar},
\[
\supp\mu=\Spec J=\mathfrak{F}^{-1}(\{0\})
\]
 and
\[
\forall\lambda\in\Spec J,\ \mu(\{\lambda\})=\bigg(\,\sum_{n=0}^{\infty}P_{n}(\lambda)^{2}\bigg)^{-1}.
\]
In view of (\ref{eq:sum-Phin}), it suffices to observe that for every
$\lambda\in\mathfrak{F}^{-1}(\{0\})$ and all $n\geq0$, $P_{n}(\lambda)=\Phi_{n}(\lambda)/\Phi_{0}(\lambda)$.
Recall again that $\Phi_{0}(z)=\mathfrak{W}(z)$. \end{proof}

\begin{proposition} The functions of the second kind defined in (\ref{eq:fce-2nd-def})
can be expressed as
\begin{equation}
\forall n\geq0,\ \,w_{n}(z)=\frac{\Phi_{n}(z)}{\mathfrak{F}(z)}.\label{eq:fce-2nd}
\end{equation}
Here again, $\{\Phi_{n}(z)\}$ is the sequences defined in (\ref{eq:Phin}),
(\ref{eq:Wgt-n}). In particular, for $n=0$ we have an expression
for the Weyl function (\ref{eq:Weyl-def}),
\begin{equation}
w(z)=\frac{\mathfrak{W}(z)}{\mathfrak{F}(z)}.\label{eq:Weyl}
\end{equation}
\end{proposition}

\begin{proof} This assertion is a direct consequence of the following
well known fact (for instance, it is in principle contained in \cite{Akhiezer},
using somewhat different terminology it can be found in \cite[Chap. 2]{Teschl},
a detailed discussion is also provided in \cite{Stovicek-JAT}): in
the Hamburger determinate case and for every $z\in\varrho(J)$, $\{w_{n}(z)\}$
is the unique square summable sequence satisfying
\begin{eqnarray*}
 &  & \alpha_{0}w_{1}(z)+\left(\beta_{0}-z\right)w_{0}(z)=1,\\
 &  & \alpha_{n}w_{n+1}(z)+(\beta_{n}-z)w_{n}(z)+\alpha_{n-1}w_{n-1}(z)=0\ \,\text{for}\ n\geq1
\end{eqnarray*}
The sequence $\{\Phi_{n}(z)\}$ is square summable, see (\ref{eq:Phin-estim}),
and a comparison with (\ref{eq:Phi-recurr}), (\ref{eq:Phi-recurr0})
leads to (\ref{eq:fce-2nd}). \end{proof}

\begin{remark}\label{rem:mass} Let us point out that there is an
alternative way how to derive the formula for the masses of atoms
of the discrete measure $\mu$ .We know that $\mu$ is supported on
the spectrum of $J$ which consists of simple eigenvalues $\lambda_{j}$,
$j\geq0$. On the other hand, the Weyl function equals the Stieltjes
transformation of $\mu$. Combining (\ref{eq:Weyl-def}) and (\ref{eq:Weyl})
we have
\begin{equation}
\frac{\mathfrak{W}(z)}{\mathfrak{F}(z)}=\int\frac{\text{d}\mu(\lambda)}{\lambda-z}=\sum_{j=0}^{\infty}\frac{\mu_{j}}{\lambda_{j}-z}\label{eq:Weyl-int}
\end{equation}
where $\mu_{j}:=\mu(\{\lambda_{j}\})$. From here we immediately obtain
(\ref{eq:mass-muj}). \end{remark}

\begin{proof}[Proof of Theorem \ref{thm:J}] Proposition \ref{thm:J-trclass},
Remark \ref{rem:tr-Jinv}, Proposition \ref{thm:Fchar} and Proposition
\ref{thm:Phi-eigenvec} jointly imply Theorem \ref{thm:J}. \end{proof}

\begin{proof}[Proof of Theorem \ref{thm:P}] Proposition \ref{thm:Pn},
Proposition \ref{thm:Fchar}, Proposition \ref{thm:J-trclass} and
Corollary \ref{thm:mass} (or Remark \ref{rem:mass}) jointly imply
Theorem \ref{thm:P}. \end{proof}

\section{Auxiliary summation identities\label{sec:Auxiliary}}

Everywhere in what follows $q\in(0,1)$. In the sequel we are using
standard notations as far as the $q$-Pochhammer symbol and the basic
hypergeometric functions are concerned, see \cite{GasperRahman,KoekoekLeskySwarttouw}.

\begin{remark} We shall need the identity
\begin{equation}
\sum_{n=0}^{\infty}\frac{q^{rn}}{(q^{n}w;q)_{r+1}}=\frac{1}{(1-q^{r})(w;q)_{r}},\ \forall r\in\mathbb{N}.\label{eq:basic-id}
\end{equation}
The both sides are regarded as meromorphic functions in $w\in\mathbb{C}$.

The identity is a straightforward consequence of a well known formula
for\linebreak{}
$\,_{1}\phi_{0}(q^{m};;q,w)$ telling us that \cite{GasperRahman,KoekoekLeskySwarttouw}
\begin{equation}
\frac{1}{(w;q)_{m}}=\sum_{s=0}^{\infty}\frac{(q^{m};q)_{s}}{(q;q)_{s}}\,w^{s}.\label{eq:phi10}
\end{equation}
Here $w\in\mathbb{C}$, $|w|<1$, and $m\in\mathbb{Z}_{+}$. In fact,
it suffices to show (\ref{eq:basic-id}) for $|w|<1$. Then, in view
of (\ref{eq:phi10}), the LHS of (\ref{eq:basic-id}) equals
\begin{eqnarray*}
\sum_{n=0}^{\infty}q^{rn}\sum_{s=0}^{\infty}\frac{(q^{r+1};q)_{s}}{(q;q)_{s}}\,q^{ns}w^{s} & = & \sum_{s=0}^{\infty}\frac{(q^{r+1};q)_{s}}{(1-q^{r+s})(q;q)_{s}}\,w^{s}\,=\,\frac{1}{1-q^{r}}\,\sum_{s=0}^{\infty}\frac{(q^{r};q)_{s}}{(q;q)_{s}}\,w^{s}\\
 & = & \frac{1}{(1-q^{r})(w;q)_{r}}.
\end{eqnarray*}
\end{remark}

\begin{lemma}\label{thm:lemma1} For every $w\in\mathbb{C}$, $|w|<1$,
and $m\in\mathbb{N}$,
\begin{equation}
(1-q^{m})(w;q)_{m}\sum_{j=0}^{\infty}\frac{q^{(m+1)j}w^{j}}{1-q^{j+m}}\,\sum_{n=0}^{\infty}\frac{q^{(j+m+1)n}}{\left(q^{n}w;q\right){}_{m+1}}=\sum_{j=0}^{\infty}\frac{q^{mj}w^{j}}{1-q^{j+m+1}}.\label{eq:lemma1}
\end{equation}
\end{lemma}

\begin{proof} We can express the LHS of (\ref{eq:lemma1}) as
\begin{eqnarray*}
 &  & \frac{(w;q)_{m}}{w}\sum_{j=0}^{\infty}\frac{q^{(m+1)j}w^{j}}{1-q^{j+m}}\,\sum_{n=0}^{\infty}q^{(j+m)n}\bigg(\frac{1}{(q^{n}w;q)_{m}}-\frac{1}{(q^{n+1}w;q)_{m}}\bigg)\\
 &  & =\,\frac{(w;q)_{m}}{w}\sum_{j=0}^{\infty}\frac{q^{(m+1)j}w^{j}}{1-q^{j+m}}\bigg(\frac{1}{(w;q)_{m}}+\sum_{n=0}^{\infty}\frac{(q^{j+m}-1)q^{(j+m)n}}{(q^{n+1}w;q)_{m}}\bigg)\\
 &  & =\,\frac{1}{w}\sum_{j=0}^{\infty}\frac{q^{(m+1)j}w^{j}}{1-q^{j+m}}-\frac{(w;q)_{m}}{w}\sum_{n=0}^{\infty}\frac{q^{mn}}{(q^{n+1}w;q)_{m+1}}.
\end{eqnarray*}
Referring to (\ref{eq:basic-id}), this expression can be further
simplified and we obtain
\begin{eqnarray*}
\frac{1}{w}\sum_{j=0}^{\infty}\frac{q^{(m+1)j}w^{j}}{1-q^{j+m}}-\frac{(w;q)_{m}}{w(1-q^{m})(qw;q)_{m}} & = & \frac{1}{w}\sum_{j=0}^{\infty}\frac{q^{(m+1)j}w^{j}}{1-q^{j+m}}-\frac{1-w}{w\,(1-q^{m})(1-q^{m}w)}\\
 & = & \frac{1}{1-q^{m}w}+\sum_{j=0}^{\infty}\frac{q^{(m+1)(j+1)}w^{j}}{1-q^{j+m+1}}.
\end{eqnarray*}
Now one can readily check that the last expression actually equals
the RHS of (\ref{eq:lemma1}). \end{proof}

\begin{proposition}\label{thm:denom-2-2--2-1} For every $a>0$ and
$m\in\mathbb{Z}_{+}$,
\begin{eqnarray}
 &  & \sum_{0\leq n_{m}\leq n_{m-1}\leq\ldots\leq n_{0}<\infty}\frac{q^{n_{0}+\ldots+n_{m-1}+n_{m}}}{(q^{n_{m}+a};q)_{2}(q^{n_{m-1}+a+2};q)_{2}\cdots(q^{n_{1}+a+2m-2};q)_{2}(q^{n_{0}+a+2m};q)_{1}}\nonumber \\
 &  & =\,\frac{1}{(q;q)_{m}(q^{a};q)_{m}}\,\sum_{j=0}^{\infty}\frac{q^{(m+a)j}}{1-q^{j+m+1}}.\label{eq:denom-2-2--2-1}
\end{eqnarray}
\end{proposition}

\begin{remark} Note that the LHS of (\ref{eq:denom-2-2--2-1}) can
be rewritten as
\begin{eqnarray*}
 &  & \hskip-1.5em\sum_{n_{0}=0}^{\infty}\sum_{n_{1}=0}^{\infty}\ldots\sum_{n_{m}=0}^{\infty}\\
 &  & \times\,\frac{q^{n_{0}+2n_{1}+\ldots+(m+1)n_{m}}}{(q^{n_{m}+a};q)_{2}(q^{n_{m-1}+n_{m}+a+2};q)_{2}\cdots(q^{n_{1}+\ldots+n_{m}+a+2m-2};q)_{2}(q^{n_{0}+n_{1}+\ldots+n_{m}+a+2m};q)_{1}}.
\end{eqnarray*}
To see it one can simply apply in this expression the substitution
\[
n_{j}+n_{j+1}+\ldots+n_{m}=n_{j}',\ 0\leq j\leq m.
\]
\end{remark}

\begin{proof} The claim is true for $m=0$ since
\[
\sum_{n=0}^{\infty}\frac{q^{n}}{1-q^{n+a}}=\sum_{j=0}^{\infty}\frac{q^{aj}}{1-q^{j+1}}.
\]
Let us denote (in this proof only) the LHS of (\ref{eq:denom-2-2--2-1})
as $T_{m}(a)$. Note that, for $\ell\in\mathbb{Z}_{+}$,
\begin{eqnarray*}
 &  & \sum_{\ell\leq n_{m}\leq n_{m-1}\leq\ldots\leq n_{0}<\infty}\frac{q^{n_{0}+\ldots+n_{m-1}+n_{m}}}{(q^{n_{m}+a};q)_{2}(q^{n_{m-1}+a+2};q)_{2}\cdots(q^{n_{1}+a+2m-2};q)_{2}(q^{n_{0}+a+2m};q)_{1}}\\
\noalign{\smallskip} &  & \qquad\quad=\,q^{(m+1)\ell}\,T_{m}(\ell+a).
\end{eqnarray*}
Thus we get, for $m\geq1$,
\begin{eqnarray*}
T_{m}(a) & = & \sum_{n=0}^{\infty}\frac{q^{n}}{\left(q^{n+a};q\right){}_{2}}\sum_{n\leq n_{m-1}\leq\ldots\leq n_{0}<\infty}q^{n_{0}+\ldots+n_{m-1}}\\
\noalign{\smallskip} &  & \hskip7em\times\,\left((q^{n_{m-1}+a+2};q)_{2}\,\cdots\,(q^{n_{1}+a+2m-2};q)_{2}(q^{n_{0}+a+2m};q)_{1}\right)^{-1}\\
 & = & \sum_{n=0}^{\infty}\frac{q^{(m+1)n}}{\left(q^{n+a};q\right){}_{2}}\,T_{m-1}(n+a+2).
\end{eqnarray*}
Hence in order to prove the formula by mathematical induction on $m$
it suffices to show that the RHS of (\ref{eq:denom-2-2--2-1}) satisfies
the same recurrence. Thus we have to verify that, for all $m\geq1$,
\begin{eqnarray*}
 &  & \frac{1}{(q;q)_{m}(q^{a};q)_{m}}\,\text{ }\sum_{j=0}^{\infty}\frac{q^{(m+a)j}}{1-q^{j+m+1}}\\
 &  & =\,\sum_{n=0}^{\infty}\frac{q^{(m+1)n}}{(q^{n+a};q)_{2}}\,\frac{1}{(q;q)_{m-1}(q^{n+a+2};q)_{m-1}}\,\sum_{j=0}^{\infty}\frac{q^{(m+n+a+1)j}}{1-q^{j+m}}.
\end{eqnarray*}
This equation can be simplified so that it takes the form
\[
\sum_{j=0}^{\infty}\frac{q^{(m+a)j}}{1-q^{j+m+1}}=(1-q^{m})(q^{a};q)_{m}\sum_{j=0}^{\infty}\frac{q^{(m+a+1)j}}{1-q^{j+m}}\,\sum_{n=0}^{\infty}\frac{q^{(m+j+1)n}}{(q^{n+a};q)_{m+1}}.
\]
But this is a consequence of Lemma \ref{thm:lemma1} for $w=q^{a}$.
\end{proof}

\begin{proposition}\label{thm:nodenom-j1-j0} Let $m\in\mathbb{Z}_{+}$
and $c_{0},c_{1},\ldots,c_{m}>0$. Then
\begin{eqnarray}
 &  & \hskip-1em\frac{1}{(1-q)^{m}}\sum_{0\leq j_{0}\leq j_{1}\leq\ldots\leq j_{m}<\infty}q^{c_{0}j_{0}+c_{1}j_{1}+\ldots+c_{m}j_{m}}(1-q^{j_{1}-j_{0}})(1-q^{j_{2}-j_{1}})\cdots(1-q^{j_{m}-j_{m-1}})\nonumber \\
\noalign{\smallskip} &  & \hskip-1em=\,\frac{q^{c_{1}+2c_{2}+\ldots+mc_{m}}}{(q^{c_{0}+c_{1}+c_{2}+\ldots+c_{m}};q)_{1}(q^{c_{1}+c_{2}+\ldots+c_{m}};q)_{2}(q^{c_{2}+\ldots+c_{m}};q)_{2}\cdots(q^{c_{m}};q)_{2}}.\label{eq:nodenom-j1-j0}
\end{eqnarray}
\end{proposition}

\begin{proof} Let us denote (in this proof only) the LHS of (\ref{eq:nodenom-j1-j0})
as $\Psi_{m}(c_{0},c_{1},\ldots,c_{m})$. We proceed by mathematical
induction on $m$. For $m=0$ the equation is just the sum of a geometric
series. For $m\geq1$ we use the substitution
\[
j_{0}=j,\ j_{1}=j_{1}'+j,\,\ldots\,,j_{m}=j_{m}'+j,
\]
and get
\begin{eqnarray*}
 &  & \hskip-2.5em\Psi_{m}(c_{0},c_{1},\ldots,c_{m})\,=\,\frac{1}{(1-q)^{m}}\sum_{j=0}^{\infty}q^{\left(c_{0}+c_{1}+\ldots+c_{m}\right)j}\sum_{0\leq j_{1}\leq\ldots\leq j_{m}<\infty}q^{c_{1}j_{1}+\ldots+c_{m}j_{m}}\\
\noalign{\medskip} &  & \qquad\qquad\qquad\qquad\qquad\qquad\qquad\qquad\times\,(1-q^{j_{1}})(1-q^{j_{2}-j_{1}})\cdots(1-q^{j_{m}-j_{m-1}})\\
\noalign{\smallskip} &  & \hskip3em=\frac{1}{(1-q)\left(1-q^{c_{0}+c_{1}+\ldots+c_{m}}\right)}\big(\Psi_{m-1}(c_{1},\ldots,c_{m})-\Psi_{m-1}(c_{1}+1,\ldots,c_{m})\big).
\end{eqnarray*}
By the induction hypothesis, the last expression equals
\begin{eqnarray*}
 &  & \hskip-2em\frac{q^{c_{2}+\ldots+(m-1)c_{m}}}{(1-q)(1-q^{c_{0}+c_{1}+\ldots+c_{m}})(q^{c_{2}+\ldots+c_{m}};q)_{2}\cdots(q^{c_{m}};q)_{2}}\\
\noalign{\smallskip} &  & \qquad\qquad\qquad\qquad\qquad\qquad\qquad\qquad\times\bigg(\frac{1}{1-q^{c_{1}+c_{2}+\ldots+c_{m}}}-\frac{1}{1-q^{c_{1}+c_{2}+\ldots+c_{m}+1}}\bigg),
\end{eqnarray*}
and this is readily seen to be equal to the RHS of (\ref{eq:nodenom-j1-j0}).
\end{proof}

\begin{proposition}\label{thm:nodenom-j0plus1} Let $m\in\mathbb{Z}_{+}$
and $c_{0},c_{1},\ldots,c_{m}>0$. Then
\begin{eqnarray}
 &  & \hskip-2em\frac{1}{(1-q)^{m+1}}\sum_{0\leq j_{0}<j_{1}<\ldots<j_{m}<\infty}q^{c_{0}j_{0}+c_{1}j_{1}+\ldots+c_{m}j_{m}}\nonumber \\
\noalign{\smallskip} &  & \qquad\qquad\qquad\qquad\qquad\times\,(1-q^{j_{0}+1})(1-q^{j_{1}-j_{0}})(1-q^{j_{2}-j_{1}})\cdots(1-q^{j_{m}-j_{m-1}})\nonumber \\
\noalign{\smallskip} &  & \hskip-2em=\frac{q^{c_{1}+2c_{2}+\ldots+mc_{m}}}{(q^{c_{0}+c_{1}+c_{2}+\ldots+c_{m}};q)_{2}(q^{c_{1}+c_{2}+\ldots+c_{m}};q)_{2}\cdots(q^{c_{m}};q)_{2}}.\label{eq:nodenom-j0plus1}
\end{eqnarray}
\end{proposition}

\begin{proof} The proof is similar to that of Proposition \ref{thm:nodenom-j1-j0}.
We again denote the LHS of (\ref{eq:nodenom-j0plus1}) as $\Psi_{m}(c_{0},c_{1},\ldots,c_{m})$
and proceed by mathematical induction on $m$. Verification of the
equation for $m=0$ is elementary. For $m\geq1$ we use the substitution
\[
j_{0}=j,\ j_{1}=j_{1}'+j+1,\ldots,\ j_{m}=j_{m}'+j+1,
\]
and get
\begin{eqnarray*}
 &  & \hskip-2em\Psi_{m}(c_{0},c_{1},\ldots,c_{m})\,=\,\frac{q^{c_{1}+\ldots+c_{m}}}{(1-q)^{m+1}}\,\sum_{j=0}^{\infty}q^{\left(c_{0}+c_{1}+\ldots+c_{m}\right)j}(1-q^{j+1})\\
\noalign{\smallskip} &  & \qquad\qquad\quad\times\,\sum_{0\leq j_{1}\leq\ldots\leq j_{m}<\infty}q^{c_{1}j_{1}+\ldots+c_{m}j_{m}}\,(1-q^{j_{1}+1})(1-q^{j_{2}-j_{1}})\cdots(1-q^{j_{m}-j_{m-1}})\\
\noalign{\smallskip} &  & \qquad\qquad\qquad=\,\frac{q^{c_{1}+\ldots+c_{m}}}{(q^{c_{0}+c_{1}+\ldots+c_{m}};q)_{2}}\,\Psi_{m-1}(c_{1},\ldots,c_{m}).
\end{eqnarray*}
From here the induction step readily follows. \end{proof}

\begin{proposition}\label{thm:s1--sm-synchro} Let $m\in\mathbb{N}$,
$s_{1},\ldots,s_{m}\in\mathbb{N}$, and $a>0$. Then
\begin{eqnarray}
 &  & \hskip-1em\sum_{0\leq n_{m}\leq\ldots\leq n_{2}\leq n_{1}<\infty}q^{s_{1}n_{1}+\ldots+s_{m}n_{m}}\bigg((q^{n_{m}+a};q)_{s_{m}+1}(q^{n_{m-1}+s_{m}+a+1};q)_{s_{m-1}+1}\nonumber \\
 &  & \qquad\qquad\times\,(q^{n_{m-2}+s_{m-1}+s_{m}+a+2};q)_{s_{m-2}+1}\cdots\text{ }(q^{n_{1}+s_{2}+\ldots+s_{m}+a+m-1};q)_{s_{1}+1}\bigg)^{-1}\nonumber \\
 & = & \frac{1}{(1-q^{s_{1}})(1-q^{s_{1}+s_{2}})\cdots(1-q^{s_{1}+\ldots+s_{m}})(q^{a};q)_{s_{1}+\ldots+s_{m}}}.\label{eq:s1--sm-synchro}
\end{eqnarray}
\end{proposition}

\begin{remark}\label{rem:s1--sm-synchro-part} In particular,
\begin{equation}
\sum_{0\leq n_{m}\leq n_{m-1}\leq\ldots\leq n_{1}<\infty}\frac{q^{n_{1}+\ldots+n_{m-1}+n_{m}}}{(q^{n_{m}+2};q)_{2}(q^{n_{m-1}+4};q)_{2}\cdots(q^{n_{1}+2m};q)_{2}}=\frac{1}{(q;q)_{m}(q^{2};q)_{m}}.\label{eq:s1--sm-synchro-part}
\end{equation}
\end{remark}

\begin{proof} Let us denote the LHS of (\ref{eq:s1--sm-synchro})
as $S_{m}(s_{1},\ldots,s_{m};a)$. We proceed by mathematical induction
on $m$. For $m=1$ the equation reduces to (\ref{eq:basic-id}).
For $m\geq2$ we have
\begin{eqnarray*}
 &  & \hskip-1.4emS_{m}(s_{1},\ldots,s_{m};a)\\
 &  & \hskip-1.4em=\,\sum_{n=0}^{\infty}\frac{q^{(s_{1}+\ldots+s_{m-1}+s_{m})n}}{(q^{n+a};q)_{s_{m}+1}}\Bigg(\sum_{0\leq n_{m-1}\leq\ldots\leq n_{2}\leq n_{1}<\infty}q^{s_{1}n_{1}+\ldots+s_{m-1}n_{m-1}}\\
\noalign{\bigskip} &  & \qquad\qquad\qquad\qquad\quad\ \times\big((q^{n_{m-1}+s_{m}+n+a+1};q)_{s_{m-1}+1}(q^{n_{m-2}+s_{m-1}+s_{m}+n+a+2};q)_{s_{m-2}+1}\\
 &  & \qquad\qquad\qquad\qquad\qquad\quad\times\,\cdots(q^{n_{1}+s_{2}+\ldots+s_{m}+n+a+m-1};q)_{s_{1}+1}\big)^{-1}\Bigg)\\
 &  & \hskip-1.4em=\,\sum_{n=0}^{\infty}\frac{q^{(s_{1}+\ldots+s_{m-1}+s_{m})n}}{(q^{n+a};q)_{s_{m}+1}}\,S_{m-1}(s_{1},\ldots,s_{m-1};n+a+s_{m}+1).
\end{eqnarray*}
By the induction hypothesis and again by equation (\ref{eq:basic-id})
the last expression equals
\begin{eqnarray*}
 &  & \hskip-1em\frac{1}{(1-q^{s_{1}})(1-q^{s_{1}+s_{2}})\cdots(1-q^{s_{1}+\ldots+s_{m-1}})}\,\sum_{n=0}^{\infty}\frac{q^{(s_{1}+\ldots+s_{m-1}+s_{m})n}}{(q^{n+a};q)_{s_{1}+\ldots+s_{m-1}+s_{m}+1}}\\
\noalign{\smallskip} &  & \hskip-1em=\,\frac{1}{(1-q^{s_{1}})(1-q^{s_{1}+s_{2}})\cdots(1-q^{s_{1}+\ldots+s_{m-1}})(1-q^{s_{1}+\ldots+s_{m-1}+s_{m}})(q^{a};q)_{s_{1}+\ldots+s_{m-1}+s_{m}}}.
\end{eqnarray*}
This concludes the proof.. \end{proof}

\section{Modified $q$-Laguerre polynomials\label{sec:qLaguerre}}

We are going to consider a particular case of (\ref{eq:alpha-beta-n})
when
\begin{equation}
a_{n}=q^{-2(n+1)}(1-q^{n+1}),\ k=q^{1/2},\label{eq:an-k-part}
\end{equation}
and so
\begin{eqnarray}
\alpha_{n} & = & q^{-2n-3/2}(1-q^{n+1}),\label{eq:alpha-beta-part}\\
\beta_{n} & = & q^{-2n-2}(1-q^{n+1})+q^{-2n+1}(1-q^{n})=q^{-2(n+1)}\big(1-q^{n+1}+q^{3}(1-q^{n})\big).\nonumber 
\end{eqnarray}
The corresponding orthonormal polynomial sequence is again denoted
as $\{P_{n}(x)\}$.

Recall that the $q$-Laguerre polynomials are defined as follows \cite{Moak,KoekoekLeskySwarttouw}
\[
L_{n}^{(a)}(x;q):=\frac{(q^{a+1};q)_{n}}{(q;q)_{n}}\,\,_{1}\phi_{1}(q^{-n};q^{a+1};q,-q^{n+a+1}x),
\]
particularly,
\begin{equation}
L_{n}^{(0)}(x;q):=\,_{1}\phi_{1}(q^{-n};q;q,-q^{n+1}x),\ L_{n}^{(1)}(x;q):=\frac{1-q^{n+1}}{1-q}\,\,_{1}\phi_{1}(q^{-n};q^{2};q,-q^{n+2}x).\label{eq:L0-L1-phi11}
\end{equation}
It is known that
\[
\underset{q\to1}{\text{lim}}\,L_{n}^{(a)}\big((1-q)x;q\big)=L_{n}^{(a)}(x).
\]
Recall, too, that the Jackson $q$-Bessel functions of the second
kind are defined as \cite{Jackson-a,Jackson-b,KoekoekLeskySwarttouw}
\begin{equation}
J_{\nu}^{(2)}(x;q):=\frac{(q^{\nu+1};q)_{\infty}}{(q;q)_{\infty}}\left(\frac{x}{2}\right)^{\!\nu}\,_{0}\phi_{1}\bigg(\,;q^{\nu+1};q,-\frac{q^{\nu+1}x^{2}}{4}\bigg),\label{eq:q-Jnu-phi01}
\end{equation}
particularly,
\begin{equation}
J_{1}^{(2)}(x;q):=\frac{x}{2(1-q)}\,\,_{0}\phi_{1}\bigg(;q^{2};q,-\frac{q^{2}x^{2}}{4}\bigg).\label{eq:q-J1-phi01}
\end{equation}
It is known that all the roots of $J_{\nu}^{(2)}(x;q)$ for $\nu>-1$
are real and simple with the only cluster point at infinity \cite{Ismail}.

Here we propose a modification of the $q$-Laguerre polynomials though
this is done for the parameter $a=0$ only.

\begin{definition} \emph{The modified $q$-Laguerre polynomials}
are introduced by the equation
\begin{equation}
\tilde{L}_{n}(x;q):=q^{n+1}L_{n}^{(0)}(x;q)+(1-q)L_{n}^{(1)}(x;q),\ n\in\mathbb{Z}_{+}.\label{eq:Lmod-def}
\end{equation}
\end{definition}

Comparing (\ref{eq:Lmod-def}) to (\ref{eq:L0-L1-phi11}) one finds
that, for all $n\geq0$, $\tilde{L}_{n}(0;q)=1$. Clearly,
\[
\underset{q\to1}{\text{lim}}\,\tilde{L}_{n}\big((1-q)x;q\big)=L_{n}^{(0)}(x).
\]

\begin{lemma} It holds true that
\begin{equation}
q^{n}L_{n}^{(0)}(x;q)+L_{n-1}^{(1)}(x;q)-L_{n}^{(1)}(x;q)=0,\ \forall n\in\mathbb{Z}_{+}.\label{eq:L0-L1-id}
\end{equation}
\end{lemma}

\begin{proof} Verification is straightforward. It is sufficient to
use definition (\ref{eq:L0-L1-phi11}) and expand the basic hypergeometric
functions in the resulting expression into power series in $x$. \end{proof}

\begin{proposition}\label{thm:Lmod-recurr} The modified $q$-Laguerre
polynomials obey the three-term recurrence
\begin{eqnarray}
 &  & (1-q^{n+1})\tilde{L}_{n+1}(x;q)-\big(1-q^{n+1}+q^{3}(1-q^{n})\big)\tilde{L}_{n}(x;q)+q^{3}(1-q^{n})\tilde{L}_{n-1}(x;q)\nonumber \\
 &  & =\,-xq^{2n+2}\tilde{L}_{n}(x;q),\ \forall n\in\mathbb{Z}_{+}\label{eq:Lmod-recurr}
\end{eqnarray}
($\tilde{L}_{-1}(x;q)$ is undefined), and $\tilde{L}_{0}(x;q)=1$.
\end{proposition}

\begin{proof} First note that $\tilde{L}_{0}(x;q)=\tilde{L}_{0}(0;q)=1$.

The $q$-Laguerre polynomials are known to obey the three-term recurrence
\cite{KoekoekLeskySwarttouw}
\begin{eqnarray}
-q^{2n+1}xL_{n}^{(0)}(x;q) & = & (1-q^{n+1})\big(L_{n+1}^{(0)}(x;q)-L_{n}^{(0)}(x;q)\big)\nonumber \\
 &  & -\,q(1-q^{n})\big(L_{n}^{(0)}(x;q)-L_{n-1}^{(0)}(x;q)\big)\label{eq:L0-recurr}
\end{eqnarray}
and
\begin{eqnarray}
-q^{2n+2}xL_{n}^{(1)}(x;q) & = & (1-q^{n+1})\big(L_{n+1}^{(1)}(x;q)-L_{n}^{(1)}(x;q)\big)\label{eq:L1-recurr}\\
 &  & -\,q(1-q^{n+1})\big(L_{n}^{(1)}(x;q)-L_{n-1}^{(1)}(x;q)\big).\nonumber 
\end{eqnarray}
Using (\ref{eq:L0-L1-id}) one can rewrite (\ref{eq:L1-recurr}) as
\begin{eqnarray}
\hskip-1em-q^{2n+2}xL_{n}^{(1)}(x;q) & = & (1-q^{n+1})\big(L_{n+1}^{(1)}(x;q)-L_{n}^{(1)}(x;q)-q^{n+1}L_{n}^{(0)}(x;q)\big)\nonumber \\
\hskip-1em &  & -\,q^{3}(1-q^{n})\big(L_{n}^{(1)}(x;q)-L_{n-1}^{(1)}(x;q)-q^{n}L_{n}^{(0)}(x;q)\big).\label{eq:L1-recurr-2}
\end{eqnarray}
Taking an appropriate linear combination of (\ref{eq:L0-recurr})
and (\ref{eq:L1-recurr-2}) and using the defining equation (\ref{eq:Lmod-def})
one obtains (\ref{eq:Lmod-recurr}). \end{proof}

\begin{proposition}\label{thm:mod-q-Laguerre} Let $\{P_{n}(x);\,n\geq0\}$
be the orthonormal polynomial sequence defined in (\ref{eq:Pn-recurr}),
with $\alpha_{n}$, $\beta_{n}$ given in (\ref{eq:alpha-beta-part}).
Then
\begin{equation}
P_{n}(x)=(-1)^{n}q^{-n/2}\tilde{L}_{n}(x;q),\ \forall n\in\mathbb{Z}_{+}.\label{eq:Pn-qLaguerre}
\end{equation}
The Hamburger moment problem for $\{P_{n}(x)\}$ is determinate, the
corresponding orthogonality measure $\mu$ (normalized as a probability
measure) is supported on the roots of the function
\begin{equation}
\mathfrak{F}(z)=\frac{1-q}{\sqrt{z}}\,J_{1}^{(2)}(2\sqrt{z};q),\label{eq:Fgt-part-J1}
\end{equation}
with all the roots being positive. The masses of the roots satisfy
\[
\forall\lambda\in\mathfrak{F}^{-1}(\{0\}),\ \mu(\{\lambda\})=-\frac{\mathfrak{W}(\lambda)}{\mathfrak{F}'(\lambda)}
\]
where
\begin{equation}
\mathfrak{\mathfrak{W}}(z)=\frac{(1-q)q}{z}\,J_{2}^{(2)}(2\,\sqrt{qz};q)+\sum_{m=0}^{\infty}\frac{q^{(m+3)(m+1)}\,_{2}\phi_{1}(q^{m+2},q;q^{m+3};q,q^{m+2})}{(q;q)_{m}(q^{2};q)_{m+1}}\,(-z)^{m}.\label{eq:Wgt-part}
\end{equation}
\end{proposition}

\begin{remark*} Let us point out once more that the respective Weyl
function satisfies
\[
w(z):=\int\frac{\text{d}\mu(\lambda)}{\lambda-z}=\frac{\mathfrak{W}(z)}{\mathfrak{F}(z)},
\]
see (\ref{eq:Weyl}) and (\ref{eq:Weyl-int}). \end{remark*}

\begin{proof} The choice (\ref{eq:alpha-beta-part}) is covered by
Theorems \ref{thm:J} and \ref{thm:P} as a particular case. This
means, among others, that the polynomials $P_{n}(x)$ satisfy (\ref{eq:Pn}),
the orthogonality measure $\mu$ is supported on the roots of the
function $\mathfrak{F}(z)$, with all the roots being positive, the
masses of the roots satisfy (\ref{eq:mass-muj}), and functions $\mathfrak{F}(z)$
and $\mathfrak{W}(z)$ are defined in (\ref{eq:Fgt}) and (\ref{eq:Wgt}),
respectively.

Let us show (\ref{eq:Pn-qLaguerre}). Comparing (\ref{eq:Pn-recurr}),
where $\alpha_{n}$, $\beta_{n}$ are given in (\ref{eq:alpha-beta-part}),
with (\ref{eq:Lmod-recurr}) in Proposition \ref{eq:Lmod-recurr}
we see that the sequences $\{P_{n}(x)\}$ and $\{(-1)^{n}q^{-n/2}\tilde{L}_{n}(x;q)\}$
obey the same three-term recurrence as well as the same initial condition.
Hence the sequences are necessarily equal.

Let us show (\ref{eq:Fgt-part-J1}). Expressing the $q$-Bessel function
in terms of a basic hypergeometric function, see (\ref{eq:q-J1-phi01}),
one finds that equation (\ref{eq:Fgt-part-J1}) means that
\begin{equation}
\mathfrak{F}(z)=\,_{0}\phi_{1}(\,;q^{2};q,-q^{2}z).\label{eq:Fgoth-phi01}
\end{equation}
Recalling (\ref{eq:Fgt}) and making the choice (\ref{eq:an-k-part})
we have
\begin{eqnarray}
 &  & \hskip-1.3em\mathfrak{F}(z)\,=\,1+\sum_{n=1}^{\infty}(-1)^{n}\label{eq:Fgt-qLaguerre}\\
 &  & \hskip-1.3em\times\bigg(\,\sum_{0\leq j_{1}<j_{2}<\ldots<j_{n}<\infty}\frac{(1-q^{j_{1}+1})(1-q^{j_{2}-j_{1}})\cdots(1-q^{j_{n}-j_{n-1}})}{(1-q)^{n}(1-q^{j_{1}+1})(1-q^{j_{2}+1})\cdots(1-q^{j_{n}+1})}\,q^{2(j_{1}+j_{2}+\ldots+j_{n})+2n}\bigg)z^{n}.\nonumber 
\end{eqnarray}
Thus, comparing the coefficients at respective powers of $z$ on the
right-hand sides of (\ref{eq:Fgoth-phi01}) and (\ref{eq:Fgt-qLaguerre}),
one can see that (\ref{eq:Fgoth-phi01}) is equivalent to the countably
many equations, numbered by $m\in\mathbb{N}$,
\begin{eqnarray*}
 &  & \sum_{0\leq j_{1}<j_{2}<\ldots<j_{m}<\infty}\frac{(1-q^{j_{1}+1})(1-q^{j_{2}-j_{1}})\cdots(1-q^{j_{m}-j_{m-1}})}{(1-q)^{m}(1-q^{j_{1}+1})(1-q^{j_{2}+1})\cdots(1-q^{j_{m}+1})}\,q^{2(j_{1}+j_{2}+\ldots+j_{m})}\\
 &  & =\,\frac{q^{m(m-1)}}{(q;q)_{m}(q^{2};q)_{m}}\,.
\end{eqnarray*}
For a given $m\in\mathbb{N}$, the LHS here can be actually simplified
with the aid of (\ref{eq:nodenom-j0plus1}) and also (\ref{eq:s1--sm-synchro-part}),
and we obtain
\begin{eqnarray*}
 &  & \frac{1}{(1-q)^{m}}\sum_{n_{1}=0}^{\infty}\ldots\sum_{n_{m}=0}^{\infty}q^{n_{1}+n_{2}+\ldots+n_{m}}\\
\noalign{\smallskip} &  & \qquad\ \ \times\sum_{0\leq j_{1}<\ldots<j_{m}<\infty}q^{(n_{1}+2)j_{1}+\ldots+(n_{m}+2)j_{m}}(1-q^{j_{1}+1})(1-q^{j_{2}-j_{1}})\ldots(1-q^{j_{m}-j_{m-1}})\\
 &  & =\,q^{m(m-1)}\sum_{n_{1}=0}^{\infty}\ldots\sum_{n_{m}=0}^{\infty}\frac{q^{n_{1}+2n_{2}+\ldots+mn_{m}}}{(q^{n_{1}+n_{2}+\ldots+n_{m}+2m};q)_{2}(q^{n_{2}+\ldots+n_{m}+2m-2};q)_{2}\cdots(q^{n_{m}+2};q)_{2}}\\
 &  & =\,q^{m(m-1)}\sum_{0\leq n_{m}\leq n_{m-1}\leq\ldots\leq n_{1}<\infty}\frac{q^{n_{1}+n_{2}+\ldots+n_{m}}}{(q^{n_{1}+2m};q)_{2}(q^{n_{2}+2m-2};q)_{2}\cdots(q^{n_{m}+2};q)_{2}}\\
 &  & =\,\frac{q^{m(m-1)}}{(q;q)_{m}(q^{2};q)_{m}}.
\end{eqnarray*}

Let us show (\ref{eq:Wgt-part}). Recalling (\ref{eq:Wgt}) and making
the choice (\ref{eq:an-k-part}) we have
\begin{eqnarray*}
 &  & \hskip-2.8em\mathfrak{W}(z)=\,q^{2}\sum_{m=0}^{\infty}(-1)^{m}\bigg(\,\sum_{0\leq j_{0}<j_{1}<j_{2}<\ldots<j_{m}}q^{3j_{0}+2j_{1}+2j_{2}+\ldots+2j_{m}}\\
\noalign{\smallskip} &  & \hskip7.8em\times\,\frac{(1-q^{j_{1}-j_{0}})(1-q^{j_{2}-j_{1}})\cdots(1-q^{j_{m}-j_{m-1}})}{(1-q^{j_{0}+1})(1-q^{j_{1}+1})\cdots(1-q^{j_{m}+1})}\bigg)\!\left(\frac{q^{2}z}{1-q}\right)^{\!m}.
\end{eqnarray*}
For a given $m\in\mathbb{Z}_{m}$ let
\begin{eqnarray*}
 &  & \hskip-1.7emX_{m}\,:=\,\frac{1}{(1-q)^{m}}\sum_{0\leq j_{0}<j_{1}<j_{2}<\ldots<j_{m}}q^{3j_{0}+2j_{1}+2j_{2}+\ldots+2j_{m}}\\
 &  & \qquad\qquad\qquad\qquad\qquad\qquad\times\,\frac{(1-q^{j_{1}-j_{0}})(1-q^{j_{2}-j_{1}})\cdots(1-q^{j_{m}-j_{m-1}})}{(1-q^{j_{0}+1})(1-q^{j_{1}+1})\cdots(1-q^{j_{m}+1})}\\
\noalign{\smallskip} &  & =\,\frac{1}{(1-q)^{m}}\sum_{n_{0}=0}^{\infty}\sum_{n_{1}=0}^{\infty}\sum_{n_{2}=0}^{\infty}\ldots\sum_{n_{m}=0}^{\infty}q^{n_{0}+n_{1}+n_{2}+\ldots+n_{m}}\,\sum_{0\leq j_{0}<j_{1}<j_{2}<\ldots<j_{m}}\\
\noalign{\smallskip} &  & \quad\times\,q^{(n_{0}+3)j_{0}+(n_{1}+2)j_{1}+(n_{2}+2)j_{2}+\ldots+(n_{m}+2)j_{m}}\,(1-q^{j_{1}-j_{0}})(1-q^{j_{2}-j_{1}})\cdots(1-q^{j_{m}-j_{m-1}}).
\end{eqnarray*}
Using formula (\ref{eq:nodenom-j1-j0}) we can compute
\begin{eqnarray*}
 &  & \hskip-1.5emX_{m}\,=\,\sum_{n_{0}=0}^{\infty}\sum_{n_{1}=0}^{\infty}\sum_{n_{2}=0}^{\infty}\ldots\sum_{n_{m}=0}^{\infty}\\
 &  & \hskip-1.5em\times\,\frac{q^{n_{0}+2n_{1}+3n_{2}+\ldots+(m+1)n_{m}+(m+1)m}}{(q^{n_{0}+n_{1}+n_{2}+\ldots+n_{m}+2m+3};q)_{1}(q^{n_{1}+n_{2}+\ldots+n_{m}+2m};q)_{2}(q^{n_{2}+\ldots+n_{m}+2m-2};q)_{2}\cdots(q^{n_{m}+2};q)_{2}}\\
\noalign{\smallskip} &  & \hskip-1.5em=\,q^{(m+1)m}\sum_{0\leq n_{m}\leq\ldots\leq n_{2}\leq n_{1}\leq n_{0}<\infty}\frac{q^{n_{0}+n_{1}+n_{2}+\ldots+n_{m}}}{(q^{n_{0}+2m+3};q)_{1}(q^{n_{1}+2m};q)_{2}(q^{n_{2}+2m-2};q)_{2}\cdots(q^{n_{m}+2};q)_{2}}.
\end{eqnarray*}
Writing
\[
\frac{q}{1-q^{n_{0}+2m+3}}=\frac{1}{1-q^{n_{0}+2m+2}}-\frac{1-q}{(q^{n_{0}+2m+3};q)_{2}}
\]
we get
\begin{eqnarray*}
 &  & X_{m}\,=\,q^{(m+1)m-1}\sum_{0\leq n_{m}\leq\ldots\leq n_{2}\leq n_{1}\leq n_{0}<\infty}\\
 &  & \qquad\qquad\qquad\qquad\qquad\ \times\,\frac{q^{n_{0}+n_{1}+n_{2}+\ldots+n_{m}}}{(q^{n_{0}+2m+2};q)_{1}(q^{n_{1}+2m};q)_{2}(q^{n_{2}+2m-2};q)_{2}\cdots(q^{n_{m}+2};q)_{2}}\\
\noalign{\smallskip} &  & \qquad\ \ -\,(1-q)q^{(m+1)m-1}\sum_{0\leq n_{m}\leq\ldots\leq n_{2}\leq n_{1}\leq n_{0}<\infty}\\
 &  & \qquad\qquad\qquad\qquad\qquad\ \times\,\frac{q^{n_{0}+n_{1}+n_{2}+\ldots+n_{m}}}{(q^{n_{0}+2m+2};q)_{2}(q^{n_{1}+2m};q)_{2}(q^{n_{2}+2m-2};q)_{2}\cdots(q^{n_{m}+2};q)_{2}}.
\end{eqnarray*}
Next we use identity (\ref{eq:denom-2-2--2-1}) and again (\ref{eq:s1--sm-synchro-part})
thus obtaining
\begin{eqnarray*}
X_{m} & = & \frac{q^{(m+1)m-1}}{(q;q)_{m}(q^{2};q)_{m}}\,\sum_{j=0}^{\infty}\frac{q^{(m+2)j}}{1-q^{j+m+1}}-\frac{(1-q)q^{(m+1)m-1}}{(q;q)_{m+1}(q^{2};q)_{m+1}}\\
 & = & \frac{q^{(m+1)^{2}}}{(q;q)_{m}(q^{2};q)_{m}}\,\sum_{j=0}^{\infty}\frac{q^{(m+2)j}}{1-q^{j+m+2}}+\frac{q^{(m+1)m}}{(q;q)_{m}(q^{2};q)_{m+1}}\,.
\end{eqnarray*}

From the last expression one can deduce that
\[
\mathfrak{W}(z)=\sum_{m=0}^{\infty}\bigg(\frac{q^{(m+3)(m+1)}}{(q;q)_{m}(q^{2};q)_{m}}\,\sum_{j=0}^{\infty}\frac{q^{(m+2)j}}{1-q^{j+m+2}}+\frac{q^{(m+2)(m+1)}}{(q;q)_{m}(q^{2};q)_{m+1}}\bigg)(-z)^{m}.
\]
Let us write $\mathfrak{W}(z)=\mathfrak{W}_{\text{I}}(z)+\mathfrak{W}_{\text{II}}(z)$
where
\[
\mathfrak{W}_{\text{I}}(z)=\frac{q^{2}}{1-q^{2}}\,\sum_{m=0}^{\infty}\frac{q^{m(m-1)}}{(q;q)_{m}(q^{3};q)_{m}}\,(-q^{4}z)^{m}=\frac{q^{2}}{1-q^{2}}\,\,_{0}\phi_{1}(\,;q^{3};q,-q^{4}z)
\]
and
\[
\mathfrak{W}_{\text{II}}(z)=\sum_{m=0}^{\infty}\frac{q^{(m+3)(m+1)}}{(q;q)_{m}(q^{2};q)_{m}}\,\sum_{j=0}^{\infty}\frac{q^{(m+2)j}}{1-q^{j+m+2}}\,(-z)^{m}.
\]
Recalling (\ref{eq:q-Jnu-phi01}) we have
\[
\mathfrak{W}_{\text{I}}(z)=\frac{(1-q)q}{z}\,J_{2}^{(2)}(2\,\sqrt{qz};q).
\]
Furthermore,
\[
\sum_{j=0}^{\infty}\frac{q^{aj}}{1-q^{j+a}}=\frac{1}{1-q^{a}}\,\sum_{j=0}^{\infty}\frac{(q^{a};q)_{j}}{(q^{a+1};q)_{j}}\,q^{aj}=\frac{\,_{2}\phi_{1}(q^{a},q;q^{a+1};q,q^{a})}{1-q^{a}},
\]
and therefore
\[
\mathfrak{W}_{\text{II}}(z)=\sum_{m=0}^{\infty}\frac{q^{(m+3)(m+1)}\,_{2}\phi_{1}(q^{m+2},q;q^{m+3};q,q^{m+2})}{(q;q)_{m}(q^{2};q)_{m+1}}\,(-z)^{m}.
\]
This concludes the proof. \end{proof}

\section*{Acknowledgments}

The author acknowledges partial support by European Regional Development
Fund Project \textquotedblleft Center for Advanced Applied Science\textquotedblright{}
No. CZ.02.1.01/0.0/0.0/16\_019/0000778.

\end{document}